\documentclass[conference]{IEEEtran}

\usepackage{graphicx}
\usepackage{color}
\usepackage{placeins}
\usepackage{float}
\usepackage{tabularx,colortbl}
\usepackage{amssymb}
\usepackage{amsthm}
\usepackage{cite}
\usepackage{amsmath}

\makeatletter                               
\newif\if@restonecol
\makeatother

\makeatletter
\newif\if@restonecol
\makeatother

\usepackage[linesnumbered,ruled,vlined]{algorithm2e}
\usepackage{algpseudocode}
\usepackage{amsmath}
\usepackage{listings}                       
\usepackage[framed,numbered,autolinebreaks,useliterate]{mcode}
\theoremstyle{plain}
\newtheorem{thm}{Theorem}[section]
\newtheorem{lemm}{Lemma}
\theoremstyle{plain}
\newtheorem{coro}{Corollary}

\IEEEoverridecommandlockouts

\begin{document}
\title{Iteratively Weighted MMSE Uplink Precoding for Cell-Free Massive MIMO}
\author{Zhe~Wang$^\ast$, Jiayi~Zhang$^\ast$, Hien Quoc Ngo$^\star$, Bo~Ai$^\dag$, and M{\'e}rouane~Debbah$^\xi$$^\ddag$\\
{\small $^\ast$School of Electronic and Information Engineering, Beijing Jiaotong University, Beijing 100044, China.}\\
{\small $^\star$Institute of Electronics, Communications and Information Technology (ECIT), Queen’s University Belfast, U.K.}\\
{\small $^\dag$State Key Laboratory of Rail Traffic Control and Safety, Beijing Jiaotong University, Beijing 100044, China.}\\
{\small $^\xi$Mohamed Ben Zayed University on Artificial Intelligence, Masdar City, Abu Dhabi, United Arab Emirates.}\\
{\small $^\ddag$Technology Innovation Institute, Masdar City, Abu Dhabi, United Arab Emirates.}\\
}

\maketitle


\begin{abstract}
In this paper, we investigate a cell-free massive MIMO system with both access points and user equipments equipped with multiple antennas over the Weichselberger Rayleigh fading channel. We study the uplink spectral efficiency (SE) based on a two-layer decoding structure with maximum ratio (MR) or local minimum mean-square error (MMSE) combining applied in the first layer and optimal large-scale fading decoding method implemented in the second layer, respectively. To maximize the weighted sum SE, an uplink precoding structure based on an Iteratively Weighted sum-MMSE (I-WMMSE) algorithm using only channel statistics is proposed. Furthermore, with MR combining applied in the first layer, we derive novel achievable SE expressions and optimal precoding structures in closed-form. Numerical results validate our proposed results and show that the I-WMMSE precoding can achieve excellent sum SE performance.
\end{abstract}

\IEEEpeerreviewmaketitle
\vspace{-0.3cm}
\section{Introduction}

As a promising technology for future wireless communication, cell-free massive MIMO (CF mMIMO) has been widely investigated to achieve uniform spectral efficiency (SE) to user equipments (UEs) and improve macro-diversity \cite{7827017,9113273,[162],9174860}. In CF mMIMO networks, a large number of access points (APs), arbitrarily distributed in a wide coverage area and connected to a central processing unit (CPU), jointly serve all UEs on the same time-frequency resource. Thanks to the prominent network topology of CF mMIMO, four signal processing structures, distinguished from levels of the mutual cooperation between all APs and the assistance from the CPU, can be implemented as \cite{[162]}. Among these signal processing structures, a two-layer decoding structure is considered as an efficient decoding technique \cite{7869024,8809413,9276421}. In the first layer, each AP estimates channels and decodes the UE data locally by applying an arbitrary combining scheme based on the local channel state information (CSI). In the second layer, all the local estimates of the UE data are gathered at the CPU where they are linearly weighted by the optimal large-scale fading decoding (LSFD) coefficient to obtain the final decoding data.

The vast majority of scientific papers on CF mMIMO focus on the scenario with single-antenna UEs. However, contemporary UEs with moderate physical sizes have already been equipped with multiple antennas so it is necessary to investigate the performance of CF mMIMO systems with multi-antenna UEs. Recent works like \cite{113,8901451,194,9424703,04962} have evaluated the scenario with multi-antenna UEs in CF mMIMO systems. But all these works are based on the assumption of independent and identically distributed (i.i.d.) Rayleigh fading channels, neglecting the spatial correlation that exists in any practical channel \cite{8809413,9276421}. A practical channel model for the scenario with multi-antenna UEs is the jointly-correlated Weichselberger model \cite{1576533,9148706}. Unlike the classic Kronecker channel, which models the spatial correlation properties at the AP-side and UE-side separately and neglects the joint correlation feature for each AP-UE pair \cite{ozcelik2003deficiencies}, the Weichselberger model not only considers the correlation features at both the AP-side and UE-side but models the joint correlation dependence between each AP-UE pair.

When UEs are equipped with multiple antennas, the uplink (UL) precoding structure can be designed to further improve the performance of systems. One popular optimization objective is to maximize the weighted sum rate (WSR). The authors in \cite{5756489} and \cite{4712693} showed the equivalence between the WSR maximization problem and the Weighted sum-Minimum Mean Square Error (WMMSE) problem in MIMO systems and propose an iteratively downlink transceiver design algorithm, which is based on iterative minimization of weighted MSE, for the WSR maximization. The authors in \cite{6302110} investigated the UL precoding structure optimization based on the iterative optimization since the WMMSE problem are not jointly convex over all optimization variables.

Motivated by the above observations, we investigate a CF mMIMO system with both multi-antenna APs and UEs over the Weichselberger channel, where a two-layer decoding structure is implemented with maximum ratio (MR) or local MMSE (L-MMSE) combining in each AP (the first layer) and the LSFD method in the CPU (the second layer). Then, an UL precoding structure based on an iteratively WMMSE (I-WMMSE) algorithm with only channel statistics is proposed to maximize WSR. Furthermore, we compute SE expressions and optimal precoding structures in novel closed-form with MR combining applied in the first layer.

\section{System Model}\label{sec:system}

We investigate a CF mMIMO system consisting of $M$ APs and $K$ UEs arbitrarily distributed in a wide coverage area. Both APs and UEs are equipped with multiple antennas, where $L$ and $N$ denote the number of antennas per AP and UE, respectively. We consider a standard block fading model, where the channel response is constant and frequency flat in a coherence block of $\tau _c$-length (channel uses). Let $\mathbf{H}_{mk}\in \mathbb{C}^{L\times N}$ denote the channel response between AP $m$ and UE $k$ and we assume $\mathbf{H}_{mk}$ are independent for different AP-UE pairs. The jointly-correlated (also known as the Weichselberger model \cite{1576533}) Rayleigh fading channel is given by 
\vspace*{-0.1cm}
\begin{equation}
\vspace*{-0.1cm}
\mathbf{H}_{mk}=\mathbf{U}_{mk,\mathrm{r}}\left( \mathbf{\tilde{\Omega}}_{mk}\odot \mathbf{H}_{mk,\mathrm{iid}} \right) \mathbf{U}_{mk,\mathrm{t}}^{H}
\end{equation}
where $\mathbf{U}_{mk,\mathrm{r}}=\left[ \mathbf{u}_{mk,\mathrm{r},1},\cdots ,\mathbf{u}_{mk,\mathrm{r},L} \right] \in \mathbb{C}^{L\times L}$ and $\mathbf{U}_{mk,\mathrm{t}}=\left[ \mathbf{u}_{mk,\mathrm{t},1},\cdots ,\mathbf{u}_{mk,\mathrm{t},N} \right] \in \mathbb{C}^{N\times N}$ are the eigenvector matrices of the one-sided correlation matrices $\mathbf{R}_{mk,\mathrm{r}}\triangleq \mathbb{E}\left[ \mathbf{H}_{mk}\mathbf{H}_{mk}^{H} \right]$ and $\mathbf{R}_{mk,\mathrm{t}}\triangleq \mathbb{E}\left[ \mathbf{H}_{mk}^{T}\mathbf{H}_{mk}^{*} \right]$, and $\mathbf{H}_{mk,\mathrm{iid}}\in \mathbb{C}^{L\times N}$ is composed of i.i.d. $\mathcal{N}_{\mathbb{C}}\left( 0,1 \right)$ random entries, respectively. Besides, $\mathbf{\Omega }_{mk}\triangleq \mathbf{\tilde{\Omega}}_{mk}\odot \mathbf{\tilde{\Omega}}_{mk}\in \mathbb{R} ^{L\times N}$ denotes the ``eigenmode coupling matrix'' with the $\left( l,n \right)$-th element $\left[ \mathbf{\Omega }_{mk} \right] _{ln}$ specifying the average amount of power coupling from $\mathbf{u}_{mk,\mathrm{r},l}$ to $\mathbf{u}_{mk,\mathrm{t},n}$. Moreover, $\mathbf{H}_{mk}$ can be formed as $\mathbf{H}_{mk}=\left[ \mathbf{h}_{mk,1},\cdots ,\mathbf{h}_{mk,N} \right] $ with $\mathbf{h}_{mk,n}\in \mathbb{C} ^L$ being the channel between AP $m$ and $n$-th antenna of UE $k$. By stacking the columns of $\mathbf{H}_{mk}$ on each other, we define $\mathbf{h}_{mk}\triangleq \mathrm{vec}\left( \mathbf{H}_{mk} \right) = [ \mathbf{h}_{mk,1}^{T},\cdots ,\mathbf{h}_{mk,N}^{T}] ^T\sim \mathcal{N} _{\mathbb{C}}\left( 0,\mathbf{R}_{mk} \right)$, where $\mathbf{R}_{mk}=\mathbb{E} \{ \mathbf{h}_{mk}\mathbf{h}_{mk}^{H} \} $ is the full correlation matrix \cite{9148706}
\begin{equation}
\mathbf{R}_{mk}=( \mathbf{U}_{mk,\mathrm{t}}^{*}\otimes \mathbf{U}_{mk,\mathrm{r}} ) \mathrm{diag}\left( \mathrm{vec}\left( \mathbf{\Omega }_{mk} \right) \right) ( \mathbf{U}_{mk,\mathrm{t}}^{*}\otimes \mathbf{U}_{mk,\mathrm{r}} ) ^H.
\end{equation}
The large-scale fading coefficient $\beta _{mk}$ can be extracted from $\mathbf{R}_{mk}$ as
$\beta _{mk}=\frac{1}{LN}\mathrm{tr}\left( \mathbf{R}_{mk} \right) =\frac{1}{LN}\left\| \mathbf{\Omega }_{mk} \right\| _1.$
\subsection{Channel Estimation}
Let $\tau _p$ and $\tau _c-\tau _p$ denote channel uses dedicated for the channel estimation and data transmission. In the phase of channel estimation, mutually orthogonal pilot matrices are constructed and $N$ mutually orthogonal pilot sequences are gathered to design a pilot matrix.
We define $\mathcal{P}_k$ as the index subset of UEs that use the same pilot matrix as UE $k$ including itself and $\mathbf{\Phi }_k$ as the pilot matrix assigned to UE $k$ with
\vspace*{-0.1cm}
\begin{equation}\notag
\begin{aligned}
\mathbf{\Phi }_{k}^{H}\mathbf{\Phi }_l=\begin{cases}
	\tau _p\mathbf{I}_N,&		\mathrm{if}\,\,k=l\\
	\mathbf{0}.&		\mathrm{if}\,\,k\ne l\\
\end{cases}
\end{aligned}
\end{equation}

When all UEs send their pilot matrices, the received signal at AP $m$ is $\mathbf{Y}_{m}^{\mathrm{p}}=\sum_{k=1}^K{\mathbf{H}_{mk}\mathbf{F}_{k,\mathrm{p}}\mathbf{\Phi }_{k}^{T}+\mathbf{N}_{m}^{\mathrm{p}}},$ where $\mathbf{F}_{k,\mathrm{p}}\in \mathbb{C} ^{N\times N}$ is the precoding matrix of UE $k$ for the phase of pilot transmission, $\mathbf{N}_{m}^{\mathrm{p}}\in \mathbb{C}^{L\times \tau _p}$ is the additive noise at AP $m$ with independent $\mathcal{N}_{\mathbb{C}}( 0,\sigma ^2 )$ elements and $\sigma ^2$ is the noise power, respectively. The pilot transmission is under the power constraint as $\mathrm{tr}( \mathbf{F}_{k,\mathrm{p}}\mathbf{F}_{k,\mathrm{p}}^{H}) \leqslant p_k$ with $p_k$ being the maximum transmit power of UE $k$. Then, AP $m$ computes the projection of $\mathbf{Y}_{mk}^{\mathrm{p}}$ onto $\mathbf{\Phi }_{k}^{*}$ as $\mathbf{Y}_{mk}^{\mathrm{p}}=\sum_{l\in \mathcal{P} _k}{\tau _p\mathbf{H}_{ml}\mathbf{F}_{l,\mathrm{p}}}+\mathbf{Q}_{m}^{\mathrm{p}},$ where $\mathbf{Q}_{m}^{\mathrm{p}}=\mathbf{N}_{m}^{\mathrm{p}}\mathbf{\Phi }_{k}^{*}$. Furthermore, by implementing vectorization operation, we have $\mathbf{y}_{mk}^{\mathrm{p}}\triangleq\mathrm{vec}\left( \mathbf{Y}_{mk}^{\mathrm{p}} \right)=\sum_{l\in \mathcal{P} _k}^{}{\tau _p\mathbf{\tilde{F}}_{l,\mathrm{p}}\mathbf{h}_{ml}+\mathbf{q}_{m}^{\mathrm{p}}}$,
where $\mathbf{\tilde{F}}_{l,\mathrm{p}}=\mathbf{F}_{l,\mathrm{p}}^{T}\otimes \mathbf{I}_L$ and $\mathbf{q}_{m}^{\mathrm{p}}=\mathrm{vec}\left( \mathbf{Q}_{m}^{\mathrm{p}} \right) $. Then, the MMSE estimation of $\mathbf{h}_{mk}$ is given by \cite{5340650} as
\vspace*{-0.1cm}
\begin{equation}
\mathbf{\hat{h}}_{mk}=\mathrm{vec}( \mathbf{\hat{H}}_{mk} ) =\mathbf{R}_{mk}\mathbf{\tilde{F}}_{k,\mathrm{p}}^{H}\mathbf{\Psi }_{mk}^{-1}\mathbf{y}_{mk}^{\mathrm{p}},
\vspace*{-0.1cm}
\end{equation}
where $\mathbf{\hat{H}}_{mk}$ is the MMSE estimation of $\mathbf{H}_{mk}$ and $\mathbf{\Psi }_{mk}=\sum\nolimits_{l\in \mathcal{P} _k}^{}{\tau _p\mathbf{\tilde{F}}_{l,\mathrm{p}}\mathbf{R}_{ml}\mathbf{\tilde{F}}_{l,\mathrm{p}}^{H}}+\sigma ^2\mathbf{I}_{LN}$. Note that the estimate $\mathbf{\hat{h}}_{mk}$ and estimation error $\mathbf{\tilde{h}}_{mk}=\mathbf{h}_{mk}-\mathbf{\hat{h}}_{mk}$ are independent random vectors distributed as $\mathbf{\hat{h}}_{mk}\sim \mathcal{N} _{\mathbb{C}}( 0,\mathbf{\hat{R}}_{mk} )$ and $\mathbf{\tilde{h}}_{mk}\sim \mathcal{N} _{\mathbb{C}}( 0,\mathbf{C}_{mk} )$, where $\mathbf{\hat{R}}_{mk}\triangleq \tau _p\mathbf{R}_{mk}\mathbf{\tilde{F}}_{k,\mathrm{p}}^{H}\mathbf{\Psi}_{mk}^{-1}\mathbf{\tilde{F}}_{k,\mathrm{p}}\mathbf{R}_{mk}$ and $\mathbf{C}_{mk}\triangleq \mathbf{R}_{mk}-\mathbf{\hat{R}}_{mk}$. We can form $\mathbf{R}_{mk}$ and $\mathbf{\hat{R}}_{mk}$ in the block structure as \cite{9148706} with $(n,i)$-th submatrix being $\mathbf{R}_{mk}^{ni}=\mathbb{E} \{ \mathbf{h}_{mk,n}\mathbf{h}_{mk,i}^{H} \}$ and $\mathbf{\hat{R}}_{mk}^{ni}=\mathbb{E} \{ \mathbf{\hat{h}}_{mk,n}\mathbf{\hat{h}}_{mk,i}^{H} \} $, respectively.\footnote{Let $\mathbf{X}^{ni}\in \mathbb{C} ^{L\times L}$ denote $(n,i)$-th submatrix of $\mathbf{X}\in \mathbb{C} ^{LN\times LN}$ in the following, unless mentioned. Notations applied in this section are defined similarly as that of \cite{04962}, such as ``$\odot$'', ``$\otimes$'' and ``$\mathrm{vec}\left( \cdot \right)$'', etc.}

\subsection{Data Transmission}
In the data transmission phase, all antennas of all UEs simultaneously send their data symbols to the APs. The received signal $\mathbf{y}_m\in \mathbb{C} ^L$ at AP $m$ is $\mathbf{y}_m=\sum_{k=1}^K{\mathbf{H}_{mk}\mathbf{s}_k}+\mathbf{n}_m$,
where $\mathbf{n}_m\sim \mathcal{N} _{\mathbb{C}}( 0,\sigma ^2\mathbf{I}_L ) $ is the independent receiver noise. The transmitted signal from UE $k$ $\mathbf{s}_k\in \mathbb{C} ^N$ can be constructed as $\mathbf{s}_k=\mathbf{F}_{k,\mathrm{u}}\mathbf{x}_k$, where $\mathbf{x}_k\sim \mathcal{N} _{\mathbb{C}}( 0,\mathbf{I}_N )$ is the data symbol of UE $k$ and $\mathbf{F}_{k,\mathrm{u}}\in \mathbb{C} ^{N\times N}$ is the precoding matrix for the phase of data transmission which should satisfy the power constraint of UE $k$ as $\mathrm{tr}( \mathbf{F}_{k,\mathrm{u}}\mathbf{F}_{k,\mathrm{u}}^{H} ) \leqslant p_k$.\footnote{Note that $\mathbf{F}_{k, \mathrm{u}}$ and $\mathbf{F}_{k, \mathrm{p}}$ in this paper are designed based on channel statistics so they are available for all APs and the CPU.} We implement a two-layer decoding structure to decode the data symbol.

In first layer, AP $m$ uses an arbitrary combining matrix $\mathbf{V}_{mk}\in \mathbb{C} ^{L\times N}$ to derive local detection $\mathbf{\tilde{x}}_{mk}=\mathbf{V}_{mk}^{H}\mathbf{y}_m$ of $\mathbf{x}_k$ as $\mathbf{\tilde{x}}_{mk}\!
=\!\mathbf{V}_{mk}^{H}\mathbf{H}_{mk}\mathbf{F}_{k,\mathrm{u}}\mathbf{x}_k+\sum_{l=1,l\ne k}^K{\mathbf{V}_{mk}^{H}\mathbf{H}_{ml}\mathbf{F}_{l,\mathrm{u}}\mathbf{x}_l}\!+\!\mathbf{V}_{mk}^{H}\mathbf{n}_m.$
Note that $\mathbf{V}_{mk}$ is designed based on the local channel estimates and one possible choice is MR combining $\mathbf{V}_{mk}=\mathbf{\hat{H}}_{mk}$. Besides, L-MMSE combining, which minimizes $\mathrm{MSE}_{mk}=\mathbb{E} \{ \| \mathbf{x}_k-\mathbf{V}_{mk}^{H}\mathbf{y}_m \| ^2 | \mathbf{\hat{H}}_{mk} \} $, is also a promising combining scheme as
\vspace*{-0.1cm}
\begin{equation}\label{LMMSE_Com}
\vspace*{-0.1cm}
\mathbf{V}_{mk}=\!\! \left( \sum_{l=1}^K{\left(\! \mathbf{\hat{H}}_{ml}\mathbf{\bar{F}}_{l,\mathrm{u}}\mathbf{\hat{H}}_{ml}^{H}+\mathbf{C}_{ml}^{\prime} \! \right)}+\sigma ^2\mathbf{I}_L \!\right) ^{-1}\!\!\mathbf{\hat{H}}_{mk}\mathbf{F}_{k,\mathrm{u}},
\end{equation}
where $\mathbf{\bar{F}}_{l,\mathrm{u}}\triangleq \mathbf{F}_{l,\mathrm{u}}\mathbf{F}_{l,\mathrm{u}}^{H}$ and $\mathbf{C}_{ml}^{\prime}=\mathbb{E} \{ \mathbf{\tilde{H}}_{ml}\mathbf{\bar{F}}_{l,\mathrm{u}}\mathbf{\tilde{H}}_{ml}^{H} \} \in \mathbb{C} ^{L\times L}$ with $\left( j,q \right) $-th element of $\mathbf{C}_{ml}^{\prime}$ being $\left[ \mathbf{C}_{ml}^{\prime} \right] _{jq}=\sum_{p_1=1}^N{\sum_{p_2=1}^N{\left[ \mathbf{\bar{F}}_l \right] _{p_2p_1}\left[ \mathbf{C}_{ml}^{p_2p_1} \right] _{jq}}}.$

Furthermore, we implement the ``LSFD'' method at the CPU \cite{[162]}. The CPU weights all the local estimates $\mathbf{\tilde{x}}_{mk}$ from all APs by the LSFD coefficient matrix 
as $\mathbf{\hat{x}}_k=\sum_{m=1}^M{\mathbf{A}_{mk}^{H}\mathbf{\tilde{x}}_{mk}}=\sum_{m=1}^M{\mathbf{A}_{mk}^{H}\mathbf{V}_{mk}^{H}\mathbf{H}_{mk}\mathbf{F}_{k,\mathrm{u}}\mathbf{x}_k}+\sum_{m=1}^M{\sum_{l=1,l\ne k}^K{\mathbf{A}_{mk}^{H}\mathbf{V}_{mk}^{H}\mathbf{H}_{ml}\mathbf{F}_{l,\mathrm{u}}\mathbf{x}_l}+}\mathbf{n}_{k}^{\prime},$
where $\mathbf{A}_{mk}\in \mathbb{C} ^{N\times N}$ is the complex LSFD coefficient matrix for AP $m$-UE $k$ and $\mathbf{n}_{k}^{\prime}=\sum_{m=1}^M{\mathbf{A}_{mk}^{H}\mathbf{V}_{mk}^{H}\mathbf{n}_m}$. Moreover, we can rewrite $\mathbf{\hat{x}}_k$ in a more compact form as
$\mathbf{\hat{x}}_k=\mathbf{A}_{k}^{H}\mathbf{G}_{kk}\mathbf{F}_{k,\mathrm{u}}\mathbf{x}_k+\sum_{l=1,l\ne k}^K{\mathbf{A}_{k}^{H}\mathbf{G}_{kl}\mathbf{F}_{l,\mathrm{u}}\mathbf{x}_l}+\mathbf{n}_{k}^{\prime},$
where $\mathbf{A}_k\triangleq [ \mathbf{A}_{1k}^{T},\cdots ,\mathbf{A}_{Mk}^{T} ] ^T\in \mathbb{C} ^{MN\times N}$ and $\mathbf{G}_{kl}\triangleq [ \mathbf{V}_{1k}^{H}\mathbf{H}_{1l};\cdots ;\mathbf{V}_{Mk}^{H}\mathbf{H}_{Ml} ] \in \mathbb{C} ^{MN\times N}$. Note that the CPU does not have the knowledge of channel estimates and is only aware of channel statistics. The conditional MSE matrix for UE $k$ $\mathbf{E}_k\triangleq \mathbb{E} \left\{ \left( \mathbf{x}_k-\mathbf{\hat{x}}_k \right) ( \mathbf{x}_k-\mathbf{\hat{x}}_k ) ^H\left| \mathbf{\Theta } \right. \right\}$ is
\vspace*{-0.1cm}
\begin{equation}\label{MSE_Matrix}
\begin{aligned}
\mathbf{E}_k&=\mathbf{I}_N-\mathbf{F}_{k,\mathrm{u}}^{H}\mathbb{E} \left\{ \mathbf{G}_{kk}^{H} \right\} \mathbf{A}_k-\mathbf{A}_{k}^{H}\mathbb{E} \left\{ \mathbf{G}_{kk} \right\} \mathbf{F}_{k,\mathrm{u}}\\
&+\mathbf{A}_{k}^{H}\left( \sum_{l=1}^K{\mathbb{E} \left\{ \mathbf{G}_{kl}\mathbf{\bar{F}}_{l,\mathrm{u}}\mathbf{G}_{kl}^{H} \right\}}+\sigma ^2\mathbf{S}_k \right) \mathbf{A}_k,
\vspace*{-0.1cm}
\end{aligned}
\end{equation}
where $\mathbf{\Theta }$ denotes all the channel statistics and $\mathbf{S}_k=\mathrm{diag}( \mathbb{E} \{ \mathbf{V}_{1k}^{H}\mathbf{V}_{1k} \} ,\cdots ,\mathbb{E} \{ \mathbf{V}_{Mk}^{H}\mathbf{V}_{Mk} \} )$. Then, we apply the classical use-and-then-forget (UatF) bound to derive the following ergodic achievable SE.

\begin{coro}
The achievable SE of UE $k$ can be written as
\vspace*{-0.1cm}
\begin{equation}\label{SE_Origin}
\begin{aligned}
\mathrm{SE}_k=\left( 1-\frac{\tau _p}{\tau _c} \right) \log _2\left| \mathbf{I}_N+\mathbf{D}_{k}^{H}\mathbf{\Sigma }_{k}^{-1}\mathbf{D}_k \right|,
\end{aligned}
\end{equation}
where $ \mathbf{\Sigma }_k=\sum_{l=1}^K{\mathbf{A}_{k}^{H}\mathbb{E} \{ \mathbf{G}_{kl}\mathbf{\bar{F}}_{l,\mathrm{u}}\mathbf{G}_{kl}^{H} \} \mathbf{A}_k}-\mathbf{D}_k\mathbf{D}_{k}^{H}+\sigma ^2\mathbf{A}_{k}^{H}\mathbf{S}_k\mathbf{A}_k$ and $\mathbf{D}_k=\mathbf{A}_{k}^{H}\mathbb{E} \{ \mathbf{G}_{kk} \} \mathbf{F}_{k,\mathrm{u}}$.
\end{coro}

\begin{proof}
The proof of \eqref{SE_Origin} follows similar steps in \cite{7500452,194} and is therefore omitted.
\end{proof}


\newcounter{mytempeqncnt}
\begin{figure*}[t]
\normalsize
\setcounter{mytempeqncnt}{\value{equation}}
\setcounter{equation}{7}
\begin{align} \label{eq:SE_max} 
\mathrm{SE}_{k}^{\mathrm{opt}}=\left( 1-\frac{\tau _p}{\tau _c} \right) \log _2 \left| \mathbf{I}_N+\mathbf{F}_{k,\mathrm{u}}^{H}\mathbb{E} \left\{ \mathbf{G}_{kk} \right\} \left( \sum_{l=1}^K{\mathbb{E} \left\{ \mathbf{G}_{kl}\mathbf{\bar{F}}_{l,\mathrm{u}}\mathbf{G}_{kl}^{H} \right\} -\mathbb{E} \left\{ \mathbf{G}_{kk} \right\} \mathbf{\bar{F}}_{k,\mathrm{u}}\mathbb{E} \left\{ \mathbf{G}_{kk}^{H} \right\} +\sigma ^2\mathbf{S}_k} \right) ^{-1}\mathbb{E} \left\{ \mathbf{G}_{kk} \right\} \mathbf{F}_{k,\mathrm{u}} \right|.
\end{align}
\setcounter{equation}{\value{mytempeqncnt}}
\hrulefill
\vspace*{-0.4cm}
\end{figure*}

\newcounter{mytempeqncnt2}
\begin{figure*}[t]
\normalsize
\setcounter{mytempeqncnt2}{\value{equation}}
\setcounter{equation}{9}
\begin{align}\label{eq:Gamma_2}\notag
\left[ \mathbf{\Gamma }_{kl,m}^{\left( 2 \right)} \right] _{nn^{\prime}}&=\sum_{i=1}^N{\sum_{i^{\prime}=1}^N{\left[ \mathbf{\bar{F}}_l \right] _{i^{\prime}i}\left\{ \mathrm{tr}\left( \mathbf{R}_{ml}^{i^{\prime}i}\mathbf{P}_{mkl,\left( 1 \right)}^{n^{\prime}n} \right) \right.}}\\
&\left. +\tau _{p}^{2}\sum_{q_1=1}^N{\sum_{q_2=1}^N{\left[ \mathrm{tr}\left( \mathbf{\tilde{P}}_{mkl,\left( 2 \right)}^{q_1n}\mathbf{\tilde{R}}_{ml}^{i^{\prime}q_2}\mathbf{\tilde{R}}_{ml}^{q_2i}\mathbf{\tilde{P}}_{mkl,\left( 2 \right)}^{n^{\prime}q_1} \right) +\mathrm{tr}\left( \mathbf{\tilde{P}}_{mkl,\left( 2 \right)}^{q_1n}\mathbf{\tilde{R}}_{ml}^{i^{\prime}q_2} \right) \mathrm{tr}\left( \mathbf{\tilde{P}}_{mkl,\left( 2 \right)}^{n^{\prime}q_2}\mathbf{\tilde{R}}_{ml}^{q_2i} \right) \right]}} \right\}
\end{align}
\setcounter{equation}{\value{mytempeqncnt2}}
\hrulefill
\vspace*{-0.4cm}
\end{figure*}

\newcounter{mytempeqncnt1}
\begin{figure*}[t]
\normalsize
\setcounter{mytempeqncnt1}{\value{equation}}
\setcounter{equation}{10}
\begin{align}\label{Closed_form_LSFD_MSE}
\begin{cases}
	\mathbf{A}_{k}^{\mathrm{opt}}=\left( \sum_{l=1}^K{\mathbf{T}_{kl,\left( 1 \right)}+\sum_{l\in \mathcal{P} _k}^{}{\mathbf{T}_{kl,\left( 2 \right)}}}+\sigma ^2\mathbf{S}_k \right) ^{-1}\mathbf{Z}_k\mathbf{F}_{k,\mathrm{u}},\\
	\mathbf{E}_{k}^{\mathrm{opt}}=\mathbf{I}_N-\mathbf{F}_{k,\mathrm{u}}^{H}\mathbf{Z}_{k}^{H}\left( \sum_{l=1}^K{\mathbf{T}_{kl,\left( 1 \right)}+\sum_{l\in \mathcal{P} _k}^{}{\mathbf{T}_{kl,\left( 2 \right)}}}+\sigma ^2\mathbf{S}_k \right) ^{-1}\mathbf{Z}_k\mathbf{F}_{k,\mathrm{u}}.\\
\end{cases}
\end{align}
\setcounter{equation}{\value{mytempeqncnt1}}
\hrulefill
\vspace*{-0.4cm}
\end{figure*}

\newcounter{mytempeqncnt3}
\begin{figure*}[t]
\normalsize
\setcounter{mytempeqncnt3}{\value{equation}}
\setcounter{equation}{14}
\begin{align}\label{F_Problem}\notag
&\underset{\left\{ \mathbf{F} \right\}}{\min}\sum_{k=1}^K{\mu _k\left[ \mathrm{tr}\left( \mathbf{W}_k\left( \mathbf{I}_N-\mathbf{F}_{k,\mathrm{u}}^{H}\mathbb{E} \left\{ \mathbf{G}_{kk}^{H} \right\} \mathbf{A}_k \right) \left( \mathbf{I}_N-\mathbf{F}_{k,\mathrm{u}}^{H}\mathbb{E} \left\{ \mathbf{G}_{kk}^{H} \right\} \mathbf{A}_k \right) ^H \right) \right]}\\
&\quad \  +\sum_{k=1}^K{\mu _k\left[ \mathrm{tr}\left( \mathbf{W}_k\mathbf{A}_{k}^{H}\left( \sum_{l\ne k}^K{\mathbb{E} \left\{ \mathbf{G}_{kl}\mathbf{\bar{F}}_{l,\mathrm{u}}\mathbf{G}_{kl}^{H} \right\}}+\sigma ^2\mathbf{S}_k \right) \mathbf{A}_k \right) \right]} \quad \mathrm{s}.\mathrm{t}. \left\| \mathbf{F}_{k,\mathrm{u}} \right\| ^2\leqslant p_k\,\,\forall k=1,\cdots ,K
\end{align}
\setcounter{equation}{\value{mytempeqncnt3}}
\hrulefill
\vspace*{-0.4cm}
\end{figure*}

\newcounter{mytempeqncnt4}
\begin{figure*}[t]
\normalsize
\setcounter{mytempeqncnt4}{\value{equation}}
\setcounter{equation}{15}
\begin{align}\label{Lagrange_Function}\notag
&f\left( \mathbf{F}_{1,\mathrm{u}},\cdots ,\mathbf{F}_{K,\mathrm{u}} \right) =\sum_{k=1}^K{\mu _k\left[ \mathrm{tr}\left( \mathbf{W}_k\left( \mathbf{I}_N-\mathbf{F}_{k,\mathrm{u}}^{H}\mathbb{E} \left\{ \mathbf{G}_{kk}^{H} \right\} \mathbf{A}_k \right) \left( \mathbf{I}_N-\mathbf{F}_{k,\mathrm{u}}^{H}\mathbb{E} \left\{ \mathbf{G}_{kk}^{H} \right\} \mathbf{A}_k \right) ^H \right) \right]}\\
&+\sum_{k=1}^K{\mu _k\left[ \mathrm{tr}\left( \mathbf{W}_k\mathbf{A}_{k}^{H}\left( \sum_{l\ne k}^K{\mathbb{E} \left\{ \mathbf{G}_{kl}\mathbf{\bar{F}}_{l,\mathrm{u}}\mathbf{G}_{kl}^{H} \right\}}+\sigma ^2\mathbf{S}_k \right) \mathbf{A}_k \right) \right]}+\sum_{k=1}^K{\lambda _k\left( \mathrm{tr}\left( \mathbf{F}_{k,\mathrm{u}}\mathbf{F}_{k,\mathrm{u}}^{H} \right) -p_k \right)}.
\end{align}
\setcounter{equation}{\value{mytempeqncnt4}}
\hrulefill
\vspace*{-0.4cm}
\end{figure*}

\newcounter{mytempeqncnt5}
\begin{figure*}[t]
\normalsize
\setcounter{mytempeqncnt5}{\value{equation}}
\setcounter{equation}{17}
\begin{align}\label{Term_gg}
\mathrm{tr}\left( \mathbf{R}_{mk}^{ni}\mathbf{P}_{mlk,\left( 1 \right)}^{p^{\prime}p} \right) +\tau _{p}^{2}\left[ \sum_{q_1=1}^N{\sum_{q_2=1}^N{\mathrm{tr}\left( \mathbf{\tilde{P}}_{mlk,\left( 2 \right)}^{q_1p}\mathbf{\tilde{R}}_{mk}^{nq_2}\mathbf{\tilde{R}}_{mk}^{q_2i}\mathbf{\tilde{P}}_{mlk,\left( 2 \right)}^{p^{\prime}q_1} \right) +\mathrm{tr}\left( \mathbf{\tilde{P}}_{mlk,\left( 2 \right)}^{q_1n}\mathbf{\tilde{R}}_{mk}^{nq_1} \right) \mathrm{tr}\left( \mathbf{\tilde{P}}_{mlk,\left( 2 \right)}^{p^{\prime}q_2}\mathbf{\tilde{R}}_{mk}^{q_2i} \right)}} \right]
\end{align}
\setcounter{equation}{\value{mytempeqncnt4}}
\hrulefill
\vspace*{-0.4cm}
\end{figure*}

\newcounter{mytempeqncnt6}
\begin{figure*}[t]
\normalsize
\setcounter{mytempeqncnt6}{\value{equation}}
\setcounter{equation}{18}
\begin{align}\label{E_Inverse}
\left( \mathbf{E}_{k}^{\mathrm{opt}} \right) ^{-1}=\mathbf{I}_N+\mathbf{F}_{k,\mathrm{u}}^{H}\mathbb{E} \left\{ \mathbf{G}_{kk}^{H} \right\} \left( \sum_{l=1}^K{\mathbb{E} \left\{ \mathbf{G}_{kl}\mathbf{\bar{F}}_{l,\mathrm{u}}\mathbf{G}_{kl}^{H} \right\}}-\mathbb{E} \left\{ \mathbf{G}_{kk} \right\} \mathbf{\bar{F}}_{k,\mathrm{u}}\mathbb{E} \left\{ \mathbf{G}_{kk}^{H} \right\} +\sigma ^2\mathbf{S}_k \right) ^{-1}\mathbb{E} \left\{ \mathbf{G}_{kk} \right\} \mathbf{F}_{k,\mathrm{u}}
\end{align}
\setcounter{equation}{\value{mytempeqncnt6}}
\hrulefill
\vspace*{-0.4cm}
\end{figure*}
Note that $\mathbf{A}_k$ can be optimized by the CPU based on channel statistics to maximize the achievable SE. Based on the theory of optimal receivers as in \cite{tse2005fundamentals}, we derive the optimal LSFD coefficient matrix maximizing the achievable SE as following corallary.
\begin{coro}
The achievable SE in \eqref{SE_Origin} is maximized by
\vspace*{-0.1cm}
\begin{equation}\label{Optimal_LSFD}
\vspace*{-0.1cm}
\mathbf{A}_{k}^{\mathrm{opt}}\!=\!\!\left( \sum_{l=1}^K{\mathbb{E} \{ \mathbf{G}_{kl}\mathbf{\bar{F}}_{l,\mathrm{u}}\mathbf{G}_{kl}^{H} \}}\!+\!\sigma ^2\mathbf{S}_k \right) ^{-1}\!\!\mathbb{E} \{ \mathbf{G}_{kk} \} \mathbf{F}_{k,\mathrm{u}},
\end{equation}
leading to the maximum value as \eqref{eq:SE_max}.
\end{coro}
In the sense that the optimal LSFD coefficient matrix in \eqref{Optimal_LSFD} can also minimize the conditional MSE of UE $k$ $\mathrm{MSE}_k=\mathrm{tr}\left( \mathbf{E}_k \right)$.  If the optimal LSFD coefficient matrix applied, the MSE matrix for UE $k$ can be written as
\vspace*{-0.1cm}
\addtocounter{equation}{1}
\begin{equation}\label{MMSE_MSE_Matrix}
\begin{aligned}
\mathbf{E}_{k}^{\mathrm{opt}}=\mathbf{I}_N-\mathbf{F}_{k,\mathrm{u}}^{H}\mathbb{E} \left\{ \mathbf{G}_{kk}^{H} \right\} \mathbf{A}_{k}^{\mathrm{opt}}.
\end{aligned}
\vspace*{-0.1cm}
\end{equation}
Furthermore, if MR combining $\mathbf{V}_{mk}=\mathbf{\hat{H}}_{mk}$ is applied, we derive closed-form SE expressions as the following theorem.
\begin{thm}\label{Th_Closed_Form}
For MR combining $\mathbf{V}_{mk}=\mathbf{\hat{H}}_{mk}$, the achievable SE can be computed in closed-form as $\mathrm{SE}_k=( 1-\frac{\tau _p}{\tau _c} ) \log _2| \mathbf{I}_N+\mathbf{D}_{k,\mathrm{c}}^{H}\mathbf{\Sigma }_{k,\mathrm{c}}^{-1}\mathbf{D}_{k,\mathrm{c}} |,$
where $\mathbf{\Sigma }_{k,\mathrm{c}}=\mathbf{A}_{k}^{H}( \sum_{l=1}^K{\mathbf{T}_{kl,( 1 )}+\sum_{l\in \mathcal{P} _k}^{}{\mathbf{T}_{kl,( 2 )}}} ) \mathbf{A}_k-\mathbf{D}_{k,\mathrm{c}}\mathbf{D}_{k,\mathrm{c}}^{H}+\sigma ^2\mathbf{A}_{k}^{H}\mathbf{S}_{k,\mathrm{c}}\mathbf{A}_k$ and $\mathbf{D}_{k,\mathrm{c}}=\mathbf{A}_{k}^{H}\mathbf{Z}_k\mathbf{F}_{k,\mathrm{u}}$, with $\mathbb{E} \{ \mathbf{G}_{kk} \} =\mathbf{Z}_k=[ \mathbf{Z}_{1k}^{T},\cdots ,\mathbf{Z}_{Mk}^{T} ] ^T$ and $\mathbf{S}_{k,\mathrm{c}}=\mathrm{diag}( \mathbf{Z}_{1k},\cdots ,\mathbf{Z}_{Mk} )$ with the $\left( n,n^{\prime} \right) $-th element of $\mathbf{Z}_{mk}\in \mathbb{C} ^{N\times N}$ being $\left[ \mathbf{Z}_{mk} \right] _{nn^{\prime}}=\mathrm{tr}( \mathbf{\hat{R}}_{mk}^{n^{\prime}n} )$. Moreover, $\mathbf{T}_{kl,\left( 1 \right)}\triangleq \mathrm{diag}( \mathbf{\Gamma }_{kl,1}^{( 1 )},\cdots ,\mathbf{\Gamma }_{kl,M}^{( 1 )} ) \in \mathbb{C} ^{MN\times MN}$ and
\vspace*{-0.1cm}
\begin{equation}\notag
\vspace*{-0.1cm}
\begin{aligned}
\mathbf{T}_{kl,\left( 2 \right)}^{mm'}=\left\{ \begin{array}{c}
	\mathbf{\Gamma }_{kl,m}^{\left( 2 \right)}-\mathbf{\Gamma }_{kl,m}^{\left( 1 \right)},   m=m'\\
	\mathbf{\Lambda }_{mkl}\mathbf{\bar{F}}_{l,\mathrm{u}}\mathbf{\Lambda }_{m'lk},  m\ne m'\\
\end{array} \right.
\end{aligned}
\end{equation}
where $\mathbf{T}_{kl,\left( 2 \right)}^{mm^{\prime}}$ denotes $\left( m,m^{\prime}\right) $-submatrix of $\mathbf{T}_{kl,\left( 2 \right)}\in \mathbb{C} ^{MN\times MN}$, the $\left( n,n^{\prime} \right) $-th element of $N\times N$-dimension complex matrices $\mathbf{\Lambda }_{mkl}$, $\mathbf{\Lambda }_{m^{\prime}lk}$, $\mathbf{\Gamma }_{kl,m}^{\left( 1 \right)}$ and $\mathbf{\Gamma }_{kl,m}^{\left( 2 \right)}$ are $[ \mathbf{\Lambda }_{mkl} ] _{nn^{\prime}}=\mathrm{tr}( \mathbf{\Xi }_{mkl}^{n^{\prime}n} ) $, $[ \mathbf{\Lambda }_{m^{\prime}lk} ] _{nn^{\prime}}=\mathrm{tr}( \mathbf{\Xi }_{m^{\prime}lk}^{n^{\prime}n} ) $, $[ \mathbf{\Gamma }_{mkl}^{( 1 )} ] _{nn^{\prime}}=\sum_{i=1}^N{\sum_{i^{\prime}=1}^N{[ \mathbf{\bar{F}}_{l,\mathrm{u}} ] _{i^{\prime}i}\mathrm{tr}( \mathbf{R}_{ml}^{i^{\prime}i}\mathbf{\hat{R}}_{mk}^{n^{\prime}n} )}}$ and $[ \mathbf{\Gamma }_{kl,m}^{( 2 )} ] _{nn^{\prime}}$ given by \eqref{eq:Gamma_2} with $\mathbf{\Xi }_{mkl}=\tau _p\mathbf{R}_{ml}\mathbf{\tilde{F}}_{l,\mathrm{p}}^{H}\mathbf{\Psi }_{mk}^{-1}\mathbf{\tilde{F}}_{k,\mathrm{p}}\mathbf{R}_{mk}$, $\mathbf{\Xi }_{m^{\prime}lk}=\tau _p\mathbf{R}_{m^{\prime}k}\mathbf{\tilde{F}}_{k,\mathrm{p}}^{H}\mathbf{\Psi }_{m^{\prime}k}^{-1}\mathbf{\tilde{F}}_{l,\mathrm{p}}\mathbf{R}_{m^{\prime}l}$, $\mathbf{P}_{mkl,( 1 )}=\tau _p\mathbf{S}_{mk}( \mathbf{\Psi }_{mk}-\tau _p\mathbf{\tilde{F}}_{l,\mathrm{p}}\mathbf{R}_{ml}\mathbf{\tilde{F}}_{l,\mathrm{p}}^{H} ) \mathbf{S}_{mk}^{H}$, $\mathbf{S}_{mk}=\mathbf{R}_{mk}\mathbf{\tilde{F}}_{k,\mathrm{p}}^{H}\mathbf{\Psi }_{mk}^{-1}$, $\mathbf{P}_{mkl,( 2 )}=\mathbf{S}_{mk}\mathbf{\tilde{F}}_{l,\mathrm{p}}\mathbf{R}_{ml}\mathbf{\tilde{F}}_{l,\mathrm{p}}^{H}\mathbf{S}_{mk}^{H}$, $\mathbf{\tilde{R}}_{ml}^{ni}$ and $\mathbf{\tilde{P}}_{mkl,( 2 )}^{ni}$ being $( n,i )$-submatrix of $\mathbf{R}_{ml}^{\frac{1}{2}}$ and $\mathbf{P}_{mkl,( 2 )}^{\frac{1}{2}}$, respectively. Furthermore, the optimal LSFD coefficient matrix in \eqref{Optimal_LSFD} and MSE matrix in \eqref{MMSE_MSE_Matrix} can also be computed in closed-form as \eqref{Closed_form_LSFD_MSE}.
\end{thm}
\begin{IEEEproof}
The proof of Theorem~\ref{Th_Closed_Form} follows similar steps in \cite{8187178,9148706,04962} and is therefore omitted.
\end{IEEEproof}


%
%
\section{Iteratively WMMSE precoding design}\label{sec:Iterative Optimization}
In this section, we focus on the design of UL precoding matrices. A popular weighted sum-rate maximization problem\footnote{Note that ``SE'' is equivalent to ``rate'' except from having one scaling factor $(\tau_c-\tau_p)/\tau_c$ . Since $\tau_c$ and $\tau_p$ are constants, so we ignore the difference between SE and rate in the optimization problem.} is investigated as\footnote{We only optimize the precoding matrices for the phase of data transmission $\mathbf{F}_{k,\mathrm{u}}$. The optimization of  $\mathbf{F}_{k,\mathrm{p}}$ is left for future research. The notation $\mathbf{F}$ is short for $\{ \mathbf{F}_{k,\mathrm{u}} \} _{k=1,\cdots ,K}$, denoting all variables $\mathbf{F}_{k,\mathrm{u}}$ with $k=1,\cdots ,K$. Similar definitions are applied for $\mathbf{A}$, $\mathbf{W}$, $\mathbf{S}$ in the following. The notation $\mathbf{G}$ denotes all $\mathbf{G}$-relevant variables, like $\mathbb{E} \{ \mathbf{G}_{kl}\mathbf{\bar{F}}_{l,\mathrm{u}}\mathbf{G}_{kl}^{H} \} $ and $\mathbb{E} \{ \mathbf{G}_{kk} \}$, etc.}
\vspace*{-0.1cm}
\addtocounter{equation}{2}
\begin{equation}\label{Sum_SE}
\begin{aligned}
\underset{\left\{ \mathbf{F},\mathbf{A},\mathbf{G},\mathbf{S} \right\}}{\max}\sum_{k=1}^K{\mu _k\mathrm{SE}_{k}} \ \mathrm{s}.\mathrm{t}. \left\| \mathbf{F}_{k,\mathrm{u}} \right\| ^2\leqslant p_k\,\forall k=1,\cdots ,K
\end{aligned}
\end{equation}
where $\mu _k$ represents the priority weight of UE $k$ and $\mathrm{SE}_{k}$ is given by \eqref{SE_Origin} with arbitrary combining structure in the first decoding layer.
As in \cite{5756489} and \cite{4712693}, the matrix-weighted sum-MSE minimization problem as
\vspace*{-0.1cm}
\begin{equation}\label{Sum_MSE}
\vspace*{-0.1cm}
\begin{aligned}
&\underset{\left\{ \mathbf{F},\mathbf{A},\mathbf{W},\mathbf{G},\mathbf{S} \right\}}{\min}\sum_{k=1}^K{\mu _k}\left[ \mathrm{tr}\left( \mathbf{W}_k\mathbf{E}_k \right) -\log _2\left| \mathbf{W}_k \right| \right]\\
&\mathrm{s}.\mathrm{t}. \left\| \mathbf{F}_{k,\mathrm{u}} \right\| ^2\leqslant p_k\,\,\forall k=1,\cdots ,K
\end{aligned}
\end{equation}
is equivalent to the weighted sum-rate maximization problem \eqref{Sum_SE}, where $\mathbf{W}_k$ is the weight matrix for UE $k$. Note that \eqref{Sum_MSE} is convex over each optimization variable $\mathbf{F}$, $\mathbf{A}$, $\mathbf{W}$, $\mathbf{G}$, $\mathbf{S}$ but is not jointly convex over all optimization variables. So we can solve \eqref{Sum_MSE} by sequentially fixing four of the five optimization variables $\mathbf{F}$, $\mathbf{A}$, $\mathbf{W}$, $\mathbf{G}$, $\mathbf{S}$ and updating the fifth.\footnote{As for $\mathbf{G}$ and $\mathbf{S}$, if L-MMSE combining structure applied, $\mathbb{E} \left\{ \mathbf{G}_{kk} \right\}$ and $\mathbf{S}_k$ are relevant to $\mathbf{F}_k$ so we should also update them. On the contrary, $\mathbb{E} \{ \mathbf{G}_{kk} \}$ and $\mathbf{S}_k$ with MR combining structure are irrelevant to $\mathbf{F}$ so we only need to update $\mathbb{E} \{ \mathbf{G}_{kl}\mathbf{\bar{F}}_{l,\mathrm{u}}\mathbf{G}_{kl}^{H} \} $.}

Note that optimal $\mathbf{W}_k$ for \eqref{Sum_MSE} is $\mathbf{W}_{k}^{\mathrm{opt}}=\mathbf{E}_{k}^{-1}$,  which can be easily derived through the first order optimality condition for $\mathbf{W}_k$. And the update of $\mathbf{A}_k$ and $\mathbf{E}_k$ are given by the optimal LSFD structure \eqref{Optimal_LSFD} and MSE matrix with optimal LSFD structure \eqref{MMSE_MSE_Matrix}. Substituting $\mathbf{A}_{k}^{\mathrm{opt}}$ and $\mathbf{W}_{k}^{\mathrm{opt}}$ for all UEs in \eqref{Sum_MSE}, we obtain the equivalent optimization problem:
\vspace*{-0.1cm}
\begin{equation}\label{SE_MSE}
\vspace*{-0.1cm}
\begin{aligned}
&\underset{\left\{ \mathbf{F},\mathbf{G},\mathbf{S} \right\}}{\min}\sum_{k=1}^K{\mu _k\log _2\left| \left( \mathbf{E}_{k}^{\mathrm{opt}} \right) ^{-1} \right|}\\
&\mathrm{s}.\mathrm{t}. \left\| \mathbf{F}_{k,\mathrm{u}} \right\| ^2\leqslant p_k\,\,\forall k=1,\cdots ,K
\end{aligned}
\end{equation}
which is a well-known relationship between $\mathbf{E}_k^{\mathrm{opt}}$ and $\mathrm{SE}_k^{\mathrm{opt}}$ and proven in Appendix~\ref{MSE_SE}. Last but not least, fixing other variables, the update of $\mathbf{F}_{k,\mathrm{u}}$ results in the optimization problem as \eqref{F_Problem}, which is a convex quadratic optimization problem. Thus, we can apply classic Lagrange multipliers methods and Karush-Kuhn-Tucker (KKT) conditions to derive an optimal solution. The Lagrange function of \eqref{F_Problem} is given by \eqref{Lagrange_Function}. By applying the first-order optimality condition of \eqref{Lagrange_Function} with respect to each $\mathbf{F}_{k,\mathrm{u}}$ and fixing other optimization variables, we obtain the optimal precoding structure as
\addtocounter{equation}{1}
\addtocounter{equation}{1}
\vspace*{-0.1cm}
\begin{equation}\label{Optimal_F}
\begin{aligned}
\mathbf{F}_{k,\mathrm{u}}^{\mathrm{opt}}=&\mu _k\left( \sum_{l=1}^K{\mu _l\mathbb{E} \left\{ \mathbf{G}_{lk}^{H}\mathbf{A}_l\mathbf{E}_{l}^{-1}\mathbf{A}_{l}^{H}\mathbf{G}_{lk} \right\}}+ \lambda _k\mathbf{I}_N \right) ^{-1}\\
&\times\mathbb{E} \left\{ \mathbf{G}_{kk}^{H} \right\} \mathbf{A}_k\mathbf{E}_{k}^{-1},
\end{aligned}
\end{equation}
where $\lambda _k\geqslant 0$ is the Lagrangian multiplier. According to the KKT condition. $\lambda _k$ and $\mathbf{F}_{k,\mathrm{u}}$ should also satisfy $\| \mathbf{F}_{k,\mathrm{u}} \| ^2\leqslant p_k$ and $\lambda _k( \| \mathbf{F}_{k,\mathrm{u}} \| ^2-p_k ) =0$ with $\lambda _k\geqslant 0$.
Let $\mathbf{F}_{k,\mathrm{u}}( \lambda _k ) $ denote the right-hand side of \eqref{Optimal_F}, when $\sum_{l=1}^K{\mu _l\mathbb{E} \{ \mathbf{G}_{lk}^{H}\mathbf{A}_l\mathbf{E}_{l}^{-1}\mathbf{A}_{l}^{H}\mathbf{G}_{lk} \}}$ is invertible and $\mathrm{tr}\left[ \mathbf{F}_{k,\mathrm{u}}( 0 ) \mathbf{F}_{k,\mathrm{u}}( 0 \right) ^H ] \leqslant p_k$, then $\mathbf{F}_{k,\mathrm{u}}^{\mathrm{opt}}=\mathbf{F}_{k,\mathrm{u}}\left( 0 \right) $, otherwise we must have $\mathrm{tr}[ \mathbf{F}_{k,\mathrm{u}}( \lambda _k ) \mathbf{F}_{k,\mathrm{u}}( \lambda _k ) ^H ] =p_k$, where $\lambda _k$ can be easily found by a one-dimensional (1-D) bisection algorithm due to the fact that $\mathrm{tr}[ \mathbf{F}_{k,\mathrm{u}}( \lambda _k ) \mathbf{F}_{k,\mathrm{u}}( \lambda _k ) ^H ] $ is a monotonically decreasing function of $\lambda _k$ \cite{5756489}. Moreover, if MR combining $\mathbf{V}_{mk}=\mathbf{\hat{H}}_{mk}$ applied, we can compute expectations in \eqref{Optimal_F} in closed-form as following theorem.
\begin{thm}\label{F_Th_Closed_Form}
For MR combining $\mathbf{V}_{mk}=\mathbf{\hat{H}}_{mk}$, we can derive $\mathbb{E} \{ \mathbf{G}_{kk}^{H} \} =\mathbf{Z}_{k}^{H}$, $\mathbf{A}_k$ and $\mathbf{E}_{k}$ as Theorem~\ref{Th_Closed_Form}. As for $\mathbb{E} \{ \mathbf{G}_{lk}^{H}\mathbf{A}_l\mathbf{E}_{l}^{-1}\mathbf{A}_{l}^{H}\mathbf{G}_{lk} \} $, by applying Lemma~\ref{Lemma1}, the $( i,n )$-th entry of it is $\mathrm{tr}( \mathbf{\bar{A}}_l\mathbb{E} \{ \mathbf{g}_{lk,n}\mathbf{g}_{lk,i}^{H} \} ) $ where $\mathbf{\bar{A}}_l\triangleq \mathbf{A}_l\mathbf{E}_{l}^{-1}\mathbf{A}_{l}^{H}$. The $[ \left( m-1 \right) N+p,\left( m^{\prime}-1 \right) N+p^{\prime}] $-th (or $[o,j]$-th briefly) entry of $\mathbb{E} \{ \mathbf{g}_{lk,n}\mathbf{g}_{lk,i}^{H} \} \in \mathbb{C} ^{MN\times MN}$
is
\vspace*{-0.1cm}
\begin{equation}\label{gg}\notag
\vspace*{-0.1cm}
\begin{aligned}
\mathbb{E} \{ \mathbf{g}_{lk,n}\mathbf{g}_{lk,i}^{H}\} _{oj}=\begin{cases}
	0, \quad l\notin \mathcal{P} _k,m\ne m^{\prime}\\
	\mathrm{tr}( \mathbf{R}_{mk}^{ni}\mathbf{\hat{R}}_{ml}^{p^{\prime}p} ) ,\quad l\notin \mathcal{P} _k,m=m^{\prime}\\
	\mathrm{tr}( \mathbf{\Xi }_{mlk}^{np} ) \mathrm{tr}( \mathbf{\Xi }_{m^{\prime}kl}^{p^{\prime}i} ) , \quad l\in \mathcal{P} _k,m\ne m^{\prime}\\
	\eqref{Term_gg}, \quad l\in \mathcal{P} _k,m=m^{\prime}\\
\end{cases}
\end{aligned}
\end{equation}
where $\mathbf{\Xi }_{mlk}=\tau _p\mathbf{R}_{mk}\mathbf{\tilde{F}}_{k,\mathrm{p}}^{H}\mathbf{\Psi }_{mk}^{-1}\mathbf{\tilde{F}}_{l,\mathrm{p}}\mathbf{R}_{ml}$, $\mathbf{\Xi }_{m^{\prime}kl}=\tau _p\mathbf{R}_{m^{\prime}l}\mathbf{\tilde{F}}_{l,\mathrm{p}}^{H}\mathbf{\Psi }_{m^{\prime}l}^{-1}\mathbf{\tilde{F}}_{k,\mathrm{p}}\mathbf{R}_{m^{\prime}k}$, $\mathbf{S}_{ml}=\mathbf{R}_{ml}\mathbf{\tilde{F}}_{l,\mathrm{p}}^{H}\mathbf{\Psi }_{ml}^{-1}$, $\mathbf{P}_{mlk,\left( 1 \right)}=\tau _p\mathbf{S}_{ml}( \mathbf{\Psi }_{ml}-\tau _p\mathbf{\tilde{F}}_{k,\mathrm{p}}\mathbf{R}_{mk}\mathbf{\tilde{F}}_{k,\mathrm{p}}^{H} ) \mathbf{S}_{ml}^{H}$ and $\mathbf{P}_{mlk,\left( 2 \right)}=\mathbf{S}_{ml}\mathbf{\tilde{F}}_{k,\mathrm{p}}\mathbf{R}_{mk}\mathbf{\tilde{F}}_{k,\mathrm{p}}^{H}\mathbf{S}_{ml}^{H}$.
\end{thm}
\vspace*{-0.1cm}
Furthermore, an iterative optimization algorithm for $\mathbf{F}_{k,\mathrm{u}}$, called ``iteratively WMMSE (I-WMMSE) algorithm", is summarized in Algorithm \ref{algo:iterative}. The convergence of Algorithm \ref{algo:iterative} is proven in \cite[Theorem 3]{5756489}. We notice that the optimal design of $\mathbf{F}_{k,\mathrm{u}}^{\mathrm{opt}}$ can only be implemented in the CPU, but relies only on channel statistics so it undoubtedly makes sense to improve SE performance. As for the complexity analysis, we ignore bisection steps for $\lambda _k$ in the complexity analysis. The per-iteration complexity of iterative optimization based on L-MMSE combining with the Monte-Carlo method, MR combining with the Monte-Carlo method and MR combining with the closed-form expressions are $\mathcal{O} \left( M^2K^2N^3N_r \right) $, $\mathcal{O} \left( M^2K^2N^3N_r+M^3KN^3 \right) $ and $\mathcal{O} \left( M^3K^2N^5 \right) $, respectively, where $N_r$ is the number of channel realizations.

\begin{algorithm}[t]
\label{algo:iterative}
\caption{I-WMMSE Algorithm for Design of the UL Precoding Matrix}
\KwIn{Channel statistics $\mathbf{\Theta }$ for all possible pairs; UE weights $\mu _k$ for all UEs;}
\KwOut{Optimal precoding matrices $\mathbf{F}_{k,\mathrm{u}}$ for all UEs;}

{\bf Initiation:} $i=0$, $\mathbf{F}_{k,\mathrm{u}}^{\left( 0 \right)}$ and $R^{\left( 0 \right)}=\sum_{k=1}^K{\mu _k\mathrm{SE}_{k}^{\left( 0 \right)}}$ for all UEs; maximum iteration number $I_{\max}$ and threshold $\varepsilon $;\\

\Repeat(){$\left| R^{\left( i \right)}-R^{\left( i-1 \right)} \right|/{R^{\left( i-1 \right)}}\leqslant \varepsilon $ or $i\geqslant I_{\max}$}
{
$i=i+1$\\
Update channel statistics $\mathbf{\Theta }^{\left( i \right)}$, such as $\mathbb{E} \{ \mathbf{G}_{kk}^{\left( i \right)} \} $, $\mathbb{E} \{ \mathbf{G}_{kl}^{\left( i \right)}\mathbf{\bar{F}}_{l,\mathrm{u}}^{\left( i-1 \right)}( \mathbf{G}_{kl}^{\left( i \right)} ) ^H \}$ and $\mathbf{S}_{k}^{\left( i \right)}$;\\
Update optimal LSFD matrix $\mathbf{A}_{k}^{\left( i \right)}$ with $\mathbf{F}_{l,\mathrm{u}}^{\left( i-1 \right)}$ and $\mathbf{\Theta }^{\left( i \right)}$ based on \eqref{Optimal_LSFD};\\
Update optimal MSE matrix $\mathbf{E}_{k}^{\left( i \right)}$ with $\mathbf{F}_{l,\mathrm{u}}^{\left( i-1 \right)}$, $\mathbf{A}_{k}^{\left( i \right)}$ and $\mathbb{E} \{ \mathbf{G}_{kk}^{\left( i \right)} \} $ based on \eqref{MMSE_MSE_Matrix} and update $\mathbf{W}_{k}^{\left( i \right)}$;\\
Update optimal precoding matrix $\mathbf{F}_{k}^{\left( i \right)}$ with $\mathbf{A}_{k}^{\left( i \right)}$, $\mathbf{W}_{k}^{\left( i \right)}$ and $\mathbf{\Theta }^{\left( i \right)}$ based on \eqref{Optimal_F}, where $\lambda _{k}^{\left( i \right)}$ is found by a bisection algorithm; \\
Update sum weighted rate $R^{\left( i \right)}=\sum_{k=1}^K{\mu _k\mathrm{SE}_{k}^{\left( i \right)}}$;\\
}
\end{algorithm}
\section{Numerical Results}
\begin{figure}[t]
\setlength{\abovecaptionskip}{-0.2cm}
\centering
\includegraphics[scale=0.5]{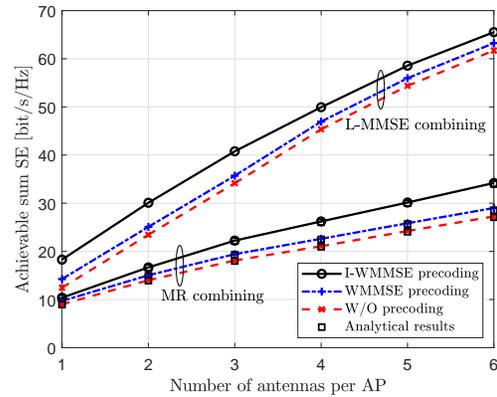}
\caption{Sum SE against the number antennas per AP $L$ over different precoding structures with $M=20$, $K=10$, and $N=4$.
\label{fig1:SE_L}}
\vspace{-0.5cm}
\end{figure}

\begin{figure}[t]
\setlength{\abovecaptionskip}{-0.1cm}
\centering
\includegraphics[scale=0.5]{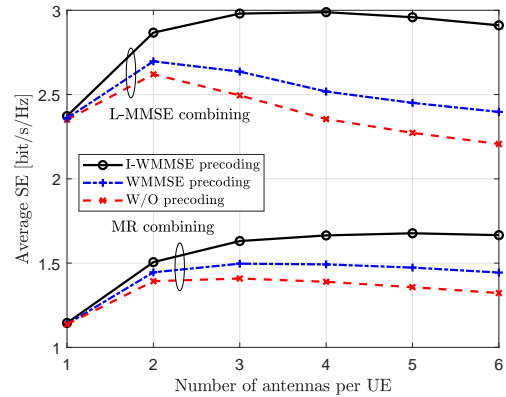}
\caption{Average SE against the number antennas per UE $N$ over different precoding structures with $M=20$, $K=10$, and $L=2$.
\label{fig2:SE_N}}
\vspace{-0.7cm}
\end{figure}

We assume all APs and UEs are uniformly distributed in a $1\times1\,\text{km}^2$ area with a wrap-around scheme. The pathloss and shadow fading are modeled similarly as \cite{8809413}. In practice, $\mathbf{U}_{mk,\mathrm{r}}$, $\mathbf{U}_{mk,\mathrm{t}}$ and $\mathbf{\Omega }_{mk}$ are estimated through measurements \cite{1576533}. But in this paper, we generate them randomly where the coupling matrix $\mathbf{\Omega }_{mk}$ consists of one strong transmit eigendirection capturing dominant power \cite{1459054}. Moreover, we have $\mathbf{F}_{k,\mathrm{p}}=\mathbf{F}_{k,\mathrm{u}}^{\left( 0 \right)}=\sqrt{\frac{p_k}{N}}\mathbf{I}_N$. As for Algorithm \ref{algo:iterative}, balancing the convergence and accuracy, we assume the maximum iteration number $I_{\max}$ and threshold $\mathrm{\varepsilon}$ are $20$ and $5\times 10^{-4}$, weights for all UEs are equal ($\mu _k=1$) without losing generality, respectively. Then, we consider communication with $20\,\text{MHz}$ bandwidth and $\sigma ^2=-94\,\text{dBm}$ noise power. All UEs transmit with $200\,\text{mW}$ power constraint. Each coherence block contains $\tau _c=200$ channel uses and $\tau _p=KN/2$.

We firstly investigate the effect of the number of antennas per AP. Fig. \ref{fig1:SE_L} shows the achievable sum SE as a function of the number of antennas per AP with L-MMSE or MR combining and ``I-WMMSE precoding", ``WMMSE precoding" or ``w/o precoding" \footnote{The ``WMMSE precoding" and ``w/o precoding" scenarios denote that precoding matrices generated by the I-WMMSE algorithm with only single iteration and identity precoding matrices $\mathbf{F}_{k,\mathrm{u}}=\sqrt{\frac{p_k}{N}}\mathbf{I}_N$ are implemented without optimization, respectively.}. We notice that I-WMMSE precoding proposed is an efficient structure to improve the achievable sum SE, even with only single iteration. With MR combining, markers ``$\circ$"  generated by analytical results overlap with the curves generated by simulations, respectively, validating our derived closed-form expressions. Moreover, the performance gap between the I-WMMSE and w/o precoding with L-MMSE combining becomes smaller with the increase of $L$, such as $46.75\%$ and $6.17\%$ SE improvement with $L=1$ and $L=6$, respectively, which implies that L-MMSE combining can use all antennas on each AP to suppress interference and achieve excellent SE performance even without any precoding structure.

Fig. \ref{fig2:SE_N} investigates the average SE as a function of the number of antennas per UE. Note that additional UE antennas may give rise to the SE degradation in the case without precoding structures \cite{194}. With the implementation of I-WMMSE precoding, we notice that UEs can make full use of multiple antennas and achieve excellent SE performance even with large $N$. Moreover, the performance gap between the I-WMMSE precoding and w/o precoding becomes larger with the increase of $N$, such as $9.43\%$ and $31.91\%$ SE improvement with $N=2$ and $N=6$, respectively, over L-MMSE combining, implying that the proposed I-WMMSE precoding is an efficient structure to improve the average SE performance especially with large $N$.

\section{Conclusion}\label{sec:conclusion}

We consider a CF mMIMO system with both APs and UEs equipped with multiple antennas over the  Weichselberger Rayleigh fading channel. A two-layer decoding structure is implemented with MR or L-MMSE combining in each AP (the first layer) and the LSFD method in the CPU (the second layer). Moreover, an UL precoding structure based on an iteratively WMMSE algorithm with only channel statistics is proposed to maximize WSR. Furthermore, we compute achievable SE expressions and optimal precoding structures in novel closed-form with MR combining in the first layer. Finally, numerical results validate our derived closed-form SE expressions and show the I-WMMSE precoding can achieve excellent sum SE performance.

\begin{appendices}
\section{A useful Lemma}
\begin{lemm}\label{Lemma1}
Let $\mathbf{X}\in \mathbb{C} ^{M\times N}$ be a random matrix and $\mathbf{Y}$ is a deterministic $M\times M$ matrix. So $(n,i)$-th element of $\mathbb{E} \{ \mathbf{X}^H\mathbf{YX} \}$ is $\mathrm{tr}\left( \mathbf{Y}\cdot \mathbb{E} \left\{ \mathbf{x}_i\mathbf{x}_{n}^{H} \right\} \right)$
where $\mathbf{x}_i$ and $\mathbf{x}_n$ is the $i$-th and $n$-th column of $\mathbf{X}$, respectively.
\end{lemm}

\section{Proof of \eqref{SE_MSE}}\label{MSE_SE}
The conditional MSE matrix for UE $k$ can be written as \eqref{MSE_Matrix}. Based on \cite{04962}, we prove that \eqref{Optimal_LSFD} can also minimize $\mathrm{MSE}_k=\mathrm{tr}\left( \mathbf{E}_k \right)$. With \eqref{Optimal_LSFD} implemented, $\mathbf{E}_k $ is given by \eqref{MMSE_MSE_Matrix}. Then, by applying \cite[Lemma B.3]{8187178}, we derive \eqref{E_Inverse} with $\mathbf{A}=\mathbf{I}_N$, $\mathbf{B}=-\mathbf{F}_{k,\mathrm{u}}^{H}\mathbb{E} \{ \mathbf{G}_{kk}^{H} \}$, $\mathbf{C}=( \sum_{l=1}^K{\mathbb{E} \{ \mathbf{G}_{kl}\mathbf{\bar{F}}_{l,\mathrm{u}}\mathbf{G}_{kl}^{H} \}}+\sigma ^2\mathbf{S}_k ) ^{-1}$ and $\mathbf{D}=\mathbb{E}\{ \mathbf{G}_{kk} \} \mathbf{F}_{k,\mathrm{u}}$, respectively. So we show the equivalence between $\mathrm{SE}_{k}^{\mathrm{opt}}$ and $\log _2| ( \mathbf{E}_{k}^{\mathrm{opt}} ) ^{-1} |$ except from having a constant scaling factor $( 1-{\tau _p}/{\tau _c} )$.

\end{appendices}

\bibliographystyle{IEEEtran}
\bibliography{IEEEabrv,Ref}
\end{document}